\documentclass[conference,letterpaper]{IEEEtran}

\usepackage[utf8]{inputenc} 
\usepackage[T1]{fontenc}
\usepackage{url}
\usepackage{ifthen}
\usepackage{cite}
\usepackage[cmex10]{amsmath} % Use the [cmex10] option to ensure complicance
\usepackage{bm,bbm,amssymb,mathrsfs,amsthm,mathtools,enumitem}

\usepackage{xcolor}

\newtheorem{theorem}{Theorem}

\theoremstyle{definition}
\newtheorem{definition}{Definition}

\newtheorem{lemma}{Lemma}
\newtheorem{corollary}{Corollary}

\newtheorem{assumption}{Assumption}

\newtheorem{remark}{Remark}

\newcommand{\R}{{\mathbb R}}

\newcommand{\E}{{\mathbb E}}

\newcommand{\Py}{{\mathbb P}}

\newcommand{\vct}[1]{{#1}}

\newcommand{\rv}[1]{{{#1}}}

\DeclareMathOperator*{\argmin}{arg\,min}

\begin{document}
\title{Testing for Anomalies: Active Strategies and Non-asymptotic Analysis} 

% %%% Single author, or several authors with same affiliation:
% \author{%
%   \IEEEauthorblockN{Stefan M.~Moser}
%   \IEEEauthorblockA{ETH Zürich\\
%                     ISI (D-ITET)\\
%                     CH-8092 Zürich, Switzerland\\
%                     Email: moser@isi.ee.ethz.ch}
% }

%%% Several authors with up to three affiliations:
%\author{%
%  \IEEEauthorblockN{Stefan M.~Moser}
%  \IEEEauthorblockA{ETH Zürich\\
%                    ISI (D-ITET), ETH Zentrum\\
%                    CH-8092 Zürich, Switzerland\\
%                    Email: moser@isi.ee.ethz.ch}
%  \and
%  \IEEEauthorblockN{Albus Dumbledore and Harry Potter}
%  \IEEEauthorblockA{Hogwarts School of Witchcraft and Wizardry\\
%                    Hogwarts Castle\\ 
%                    1714 Hogsmeade, Scotland\\
%                    Email: \{dumbledore, potter\}@hogwarts.edu}
%}

\author{%
  \IEEEauthorblockN{Dhruva Kartik, Ashutosh Nayyar and Urbashi Mitra}
  \IEEEauthorblockA{Ming Hsieh Department of Electrical and Computer Engineering\\
                    University of Southern California, Los Angeles, CA, USA\\ 
                    %3740 McClintock Ave, Los Angeles CA 90089\\
                    Email: \{mokhasun, ashutosh.nayyar, ubli\}@usc.edu}
}
%%% Many authors with many affiliations:
% \author{%
%   \IEEEauthorblockN{Albus Dumbledore\IEEEauthorrefmark{1},
%                     Olympe Maxime\IEEEauthorrefmark{2},
%                     Stefan M.~Moser\IEEEauthorrefmark{3}\IEEEauthorrefmark{4},
%                     and Harry Potter\IEEEauthorrefmark{1}}
%   \IEEEauthorblockA{\IEEEauthorrefmark{1}%
%                     Hogwarts School of Witchcraft and Wizardry,
%                     1714 Hogsmeade, Scotland,
%                     \{dumbledore, potter\}@hogwarts.edu}
%   \IEEEauthorblockA{\IEEEauthorrefmark{2}%
%                     Beauxbatons Academy of Magic,
%                     1290 Pyrénées, France,
%                     maxime@beauxbatons.edu}
%   \IEEEauthorblockA{\IEEEauthorrefmark{3}%
%                     ETH Zürich, ISI (D-ITET), ETH Zentrum, 
%                     CH-8092 Zürich, Switzerland,
%                     moser@isi.ee.ethz.ch}
%   \IEEEauthorblockA{\IEEEauthorrefmark{4}%
%                     National Chiao Tung University (NCTU), 
%                     Hsinchu, Taiwan,
%                     moser@isi.ee.ethz.ch}
% }
\thanks{asas}

\maketitle

%%%%%%
%% Abstract: 
%% If your paper is eligible for the student paper award, please add
%% the comment "THIS PAPER IS ELIGIBLE FOR THE STUDENT PAPER
%% AWARD." as a first line in the abstract. 
%% For the final version of the accepted paper, please do not forget
%% to remove this comment!
%%
\begin{abstract}
The problem of verifying whether a multi-component system has anomalies or not is addressed. Each component can be probed over time in a data-driven manner to obtain noisy observations that indicate whether the selected component is anomalous or not. The aim is to minimize the probability of incorrectly declaring the system to be free of anomalies while ensuring that the probability of correctly declaring it to be safe is sufficiently large. This problem is modeled as an active hypothesis testing problem in the Neyman-Pearson setting. Component-selection and inference strategies are designed and analyzed in the non-asymptotic regime. For a specific class of homogeneous problems, stronger (with respect to prior work) non-asymptotic converse and achievability bounds are provided.
\end{abstract}

%% The paper must be self-contained. However, if you are referring to
%% a full version for checking certain proofs, please provide the
%% publically accessible location below.  If the paper is completely
%% self-contained, you can remove the following line from your
%% submission.
%\textit{A full version of this paper is accessible at:}
%\url{https://sites.google.com/view/dhruvakartikmokhasunavisu/research} 

\section{Introduction}
Consider a system with multiple components. Each of these components may be anomalous and we would like to test whether the system has anomalies or not. When the system does not have anomalies, we say that the system is \emph{safe}, and we say that the system is \emph{unsafe} otherwise. Each component can be probed to receive a noisy observation that indicates whether or not the component is anomalous. The components can be probed sequentially in a data-driven manner. The goal is to judiciously probe the components and reliably infer whether the system is safe or unsafe.

 Such controlled-sensing problems can be modeled and analyzed in the framework of \emph{active hypothesis testing} \cite{chernoff1959sequential,nitinawarat2013controlled,naghshvar2013active}. In this paper, we consider a setting in which we are allowed to probe the components for a fixed number of times. Thus, the time-horizon of the sensing process is fixed. We are interested in two types of probabilities: correct-verification probability and incorrect-verification probability. The correct-verification event denotes that the agent inferred that the system was safe when it was indeed safe, and the incorrect-verification event denotes that the agent inferred that the system was safe when it had an anomaly. We aim to design a sensor-selection strategy and an inference strategy for the agent that minimizes the incorrect-verification probability while ensuring that the correct-verification probability is sufficiently high. For simplicity, we assume that at most one component can be anomalous in the system. We make a distinction between \emph{homogeneous} and \emph{heterogeneous} systems. In homogeneous systems, when a component is probed, the statistics of the observations depend only whether or not the component is anomalous and \emph{not} on the component's index. This may not be true in heterogeneous systems.
 
There is a plethora of works on active hypothesis testing and anomaly detection problems. These works focus largely on the asymptotic aspects whereas our goal in this paper is to focus on the non-asymptotic aspects. More specifically, we would like to generalize the strong finite-block length bounds for Neyman-Pearson type hypothesis testing in \cite{polyathesis} to a general active hypothesis testing setting. We believe this is a step forward in that direction.

Our main contributions in this paper can be summarized as follows:
 We construct a fixed-horizon Neyman-Pearson type formulation for active anomaly verification.
 For this model, we derive asymptotically optimal error rates.
 For a homogeneous system, we derive strong non-asymptotic converse bounds.
 We construct deterministic strategies that achieve this bound in a strong sense (up to an additive logarithmic term) in the non-asymptotic regime and thus, can be considered to be second-order optimal.
 Classical approaches \cite{chernoff1959sequential,nitinawarat2013controlled,naghshvar2013active} suggest an open-loop randomized component selection strategy which is asymptotically optimal. However, this strategy is not efficient in the non-asymptotic regime. 
 Our analysis herein sheds light on how one might construct second-order optimal strategies for more general active hypothesis testing problems.

Generally, we are not only interested in determining whether the system is safe or unsafe. We are also interested in finding which component is anomalous. The latter problem is known as \emph{anomaly detection}. However, we focus on the former problem for simplicity. We would like to note that the detection problem can be modeled using the symmetric formulation (P2) in \cite{kartikarxiv}, and the bounds obtained herein can then be used in that symmetric formulation.

\subsection{Related Work}\label{priorwork}
Anomaly detection and verification problems are generally analyzed within the framework of active hypothesis testing. We will first provide a brief overview of active hypothesis testing.
Hypothesis testing is a long-standing problem and has been addressed in various settings. In the simplest fixed-horizon hypothesis testing setup, we have binary hypotheses and a single experiment. The inference is made based on a fixed number of i.i.d. observations obtained by repeatedly performing this experiment. In this setup, a popular formulation is the Neyman-Pearson formulation \cite{cover2012elements}. For this formulation, tight bounds on error probabilities were proved in \cite{polyathesis}. However, not much work \cite{polyanskiy2011binary} has been done to extend these bounds to the case of \emph{active} hypothesis testing. Neyman-Pearson type active hypothesis testing problems were formulated in \cite{kartik2019active,kartikarxiv} and their asymptotics were analyzed. In this paper, we take the non-asymptotic analysis of such problems a step further.

Another widely used paradigm is the sequential setting. In sequential hypothesis testing, the time horizon is not fixed and the agent can continue to perform experiments until a stopping criterion is met. The objective then is to minimize a linear combination of the expected stopping time and the Bayesian error probability.
Inspired by Wald's sequential probability ratio test (SPRT) \cite{wald1973sequential}, Chernoff first addressed the problem of \emph{active} sequential hypothesis testing in \cite{chernoff1959sequential}. This work was later generalized in \cite{bessler1960theory,nitinawarat2013controlled,naghshvar2013active}. Although our formulation has a fixed time-horizon, it is closely related to the sequential active hypothesis testing framework. Fixed-horizon formulations are useful in applications with hard time/energy constraints and the agent does not have the luxury to keep performing experiments until strong enough evidence is obtained. In contrast to all these approaches, a Gibbs sampling-based active sensing approach was also proposed in \cite{arpan}.

Some recent works in anomaly detection include \cite{huang2018active,chen2019active,tsope,gure,chao}.
All these works are in the sequential setting and focus on asymptotic optimality. Unlike these works, we address a simpler problem in which we do not need to find the anomaly but only need to decide whether or not there is an anomaly. However, we believe our non-asymptotic analysis is considerably tighter than any of these works.

\subsection{Notation}\label{notation}
Random variables are denoted by upper case letters ($X$), their realization by the corresponding lower case letter ($x$). We use calligraphic fonts to denote sets ($\mathcal{U}$). The probability simplex over a finite set $\mathcal{U}$ is denoted by $\Delta \mathcal{U}$. In general, subscripts denote time indices unless stated otherwise. For time indices $n_1\leq n_2$, $\rv{Y}_{n_1:n_2}$ denotes the collection of variables $(\rv{Y}_{n_1},\rv{Y}_{n_1+1},...,\rv{Y}_{n_2})$.
For a strategy $g$, we use $\Py^g[\cdot]$ and $\E^g[\cdot]$ to indicate that the probability and expectation depend on the choice of $g$. For an hypothesis $i$, $\E_i^g[\cdot]$ denotes the expectation conditioned on hypothesis $i$.
The \emph{cross-entropy} between two distributions $p$ and $q$ over a finite space $\mathcal{Y}$ is given by
\begin{align}
H(p,q) = -\sum_{y \in \mathcal{Y}}p(y)\log q(y).
\end{align}
The \emph{Kullback-Leibler} divergence between distributions $p$ and $q$ is given by
\begin{equation}
D(p || q) = \sum_{y \in \mathcal{Y}}p(y)\log\frac{p(y)}{q(y)}.
\end{equation}

\section{Problem Formulation}
Consider a system with multiple components. Let the set of all the components in the system be denoted by $\mathcal{U} \doteq \{1,\dots,M\}$ where $M$ is a positive integer. For simplicity, we assume that the system may have at most one anomalous component. Thus the set of hypotheses is $\mathcal{X} \doteq \{0,1,\dots, M \}$, where 0 denotes that the system does not have any anomaly whereas $j > 0$ denotes that component $j$ is anomalous. Let the random variable $X$ represent the true hypothesis which is unknown to the agent. Let the prior distribution on $X$ be $\rho_1$.

At each time $n$, the agent can select a component $U_n \in \mathcal{U}$ and obtain an observation $Y_n \in \mathcal{Y}$. This action of selecting a component and obtaining an observation will be referred to as an \emph{experiment}. The observation $Y_n$ at time $n$ is given by
\begin{equation}
Y_n = \xi(X, U_n, W_n) \doteq
\begin{cases}
\Upsilon(U_n, W_n) & \text{if } X \neq U_n\\
\bar{\Upsilon}(U_n, W_n) & \text{if } X = U_n,
\end{cases}
\end{equation}
where $\{W_n: n =1,2,\dots\}$ is a collection of mutually independent variables and $\Upsilon, \bar{\Upsilon}$ are arbitrary measurable mappings. The density associated with an observation $y$ when component $u$ is selected is denoted by $p_1^u$ if component $u$ is anomalous and $p_0^u$ otherwise. For each component $u$, the densities are with respect to a $\sigma$-finite measure $\nu$ over the observation space $\mathcal{Y}$. Thus, for a measurable set $A \subseteq \mathcal{Y}$ and and hypothesis $j \in \mathcal{X}$,
\begin{align*}
\Py[Y_n \in A \mid U_n = u, X = j] = \begin{cases}
\int_{A} p_0^u(y)d\nu(y) & \text{if } j \neq u\\
\int_{A} p_1^u(y)d\nu(y) & \text{if } j = u.
\end{cases} 
\end{align*}
The system is said to be \emph{homogeneous} if the densities $p_0^u$ and $p_1^u$ do not depend on the component $u$. Otherwise, the system is said to be \emph{heterogeneous}. Thus, in homogeneous systems, the statistics of the observation $Y$ depend only on whether the selected component is anomalous or not. For homogeneous systems, we will drop the superscript $u$ from the densities and simply refer to them as $p_0$ and $p_1$.

The total number of observations collected by the agent is fixed and is denoted by $N$. At time $n=1,2,\ldots, N$, the information available to the agent, denoted by $\rv{I}_n$, is the collection of all experiments performed and the corresponding observations up to time $n-1$, 
\begin{equation}
\rv{I}_n \doteq \{\rv{U}_{1:n-1},\rv{Y}_{1:n-1}\}.
\end{equation}
Let the collection of all possible realizations of information $\rv{I}_n$ be denoted by $\mathcal{I}_n$. At time $n$, the agent selects a distribution over the set of components $\mathcal{U}$ according to an \emph{experiment selection rule} $g_n: \mathcal{I}_n \to \Delta \mathcal{U}$ and the action $\rv{U}_n$ is randomly drawn from the distribution $g_n(\rv{I}_n)$,
%The experiment selected at time $n$ can be chosen as a function of $\rv{I}_n$. Let the \emph{strategy} used for selecting the experiment be $g_n$,
that is,
\begin{equation}
\rv{U}_n  \sim g_n(\rv{I}_n).
\end{equation}
{For a given experiment $u \in \mathcal{U}$ and information realization $\mathscr{I} \in \mathcal{I}_n$, the probability $\Py^g[U_n = u \mid I_n = \mathscr{I}]$ is denoted by $g_n(\mathscr{I} : u)$.}
The sequence  $\{g_n, n=1,\ldots,N\}$ is denoted by $g$ and  referred to as the \emph{experiment selection strategy}. Let the collection of all such experiment selection strategies be $\mathcal{G}$. 

If the system does not have any anomalies, it is considered to be safe (denoted by $0$), and if the system has an anomaly, it is considered to by unsafe (denoted by $\aleph$). After performing $N$ experiments, the agent can declare the system to be safe or unsafe. We refer to this final declaration as the  agent's \emph{inference decision} and denote it by  $\hat{\rv{X}}_N$.  Thus, the inference decision can take values in $\{0,\aleph\}$. Using the information $I_{N+1}$, the agent chooses a distribution over the set $\{0,\aleph\}$ according to an \emph{inference strategy} $f: \mathcal{I}_{N+1} \to \Delta\{0,\aleph\}$ and the inference $\hat{\rv{X}}_N$ is drawn from the distribution $f(\rv{I}_{N+1})$,  \emph{i.e}.
\begin{equation}
    \hat{\rv{X}}_{N} \sim f(\rv{I}_{N+1}).
\end{equation}
{For a given inference $\hat{x} \in \{0,\aleph\}$ and information realization $\mathscr{I} \in \mathcal{I}_{N+1}$, the probability $\Py^{f,g}[\hat{X}_N = \hat{x} \mid I_{N+1} = \mathscr{I}]$ is denoted by $f_N(\mathscr{I} : \hat{x})$.} Let the set of all inference strategies be $\mathcal{F}$.

The system is said to be \emph{correctly verified} if the agent declares the system to be safe when it was indeed safe, and it is said to be \emph{incorrectly verified} if the agent declares it to be safe when it was actually \emph{not} safe. For an experiment selection strategy $g$ and an inference strategy $f$, we define the following probabilities.
\begin{definition}
Let $\psi_N$ be the probability that the agent declares the system to be safe given that the system is indeed safe, \emph{i.e}.
\begin{align}
    \psi_N &\doteq \Py^{f,g}[\hat{\rv{X}}_{N} = 0 \mid \rv{X} = 0].\\
    %&=\sum_{j\neq i}\Py^{f,g}[\hat{\rv{X}}_{N} = j \mid \rv{X} = i] + \Py^{f,g}[\hat{\rv{X}}_{N} = \aleph \mid \rv{X} = i]
    \shortintertext{We refer to $\psi_N$ as the \emph{correct-verification probability}. Let $\phi_N$ be the probability that the agent declares the system to be safe given that the system is not safe, \emph{i.e}.}
    \phi_N &\doteq \Py^{f,g}[\hat{\rv{X}}_{N} = 0 \mid \rv{X} \neq 0].
\end{align}
We refer to $\phi_N$ as the \emph{incorrect-verification probability}.
\end{definition}
%\begin{remark}
%Note that when there are only two hypotheses and the agent is forced to decide on one of the two hypotheses, $\psi_N(1)$ and $\psi_N(2)$ are type I and type II errors, respectively. In this case, $\psi_N(1) = \phi_N(2)$ and $\psi_N(2) = \phi_N(1)$.
%\end{remark}

We are interested in designing an experiment selection strategy $g$ and an inference strategy $f$ that minimize the incorrect-verification probability $\phi_N$ subject to the constraint that the correct-verification probability $\psi_N$ is sufficiently large.
In other words, we would like to solve the following optimization problem: 
\begin{align}\tag{P1}\label{opti}
    & & & \underset{f \in \mathcal{F},g \in \mathcal{G}}{\text{inf}} & & \phi_N & \\
    \nonumber & & & \text{subject to} & & \psi_N \geq 1-\epsilon_N, & 
\end{align}
where $0 < \epsilon_N < 1$. Let the infimum value of this optimization problem be $\phi^*_N$. 
Note that this problem is always feasible because the agent can trivially satisfy the correct-verification probability constraint by always declaring the system to be safe. We refer to this problem as Problem (\ref{opti}).

\section{Main Results}
Before we state our main results, we will define some important quantities.

\begin{definition}
For an experiment $u \in \mathcal{U}$, a component $j \in \mathcal{U}$ and for an observation $y \in \mathcal{Y}$, the log-likelihood ratio (hypothesis 0 vs hypothesis $j$) is denoted by
\begin{align}\label{llrdef}
\lambda_j(u,y) \doteq \begin{cases}
\log\frac{p_0^u(y)}{p_1^u(y)} & \text{if } u = j \\
0 & \text{otherwise}.
\end{cases}
\end{align}
Also, define $D_j^u = \E_0^u[\lambda_j(u,Y)]$ where observation $Y\sim p^u_0$. To avoid trivial cases, let $D_u^u > 0$ for every $u \in \mathcal{U}$.
\end{definition}

We make the following assumption on the log-likelihood ratios which is standard in the hypothesis testing literature \cite{chernoff1959sequential,nitinawarat2013controlled}.
\begin{assumption}\label{boundedassump}
For any given experiment $u \in \mathcal{S}$ and each $k \in \{0,1\}$,
\begin{align}
\int_y \left( \log\frac{p_0^u(y)}{p_1^u(y)} \right)^2 p_k^u(y)d\nu(y) < \infty.
\end{align}
\end{assumption}

\begin{definition}[Max-min Kullback-Leibler Divergence]
\label{kldef}
Define
\begin{align}
\label{maxmin}D^* &\doteq \max_{\vct{\alpha} \in \Delta\mathcal{U}} \min_{j \in \mathcal{U}} \sum_{u \in \mathcal{U}}\alpha(u) D_j^u\\
\label{minmax}&=  \min_{\vct{\beta} \in \Delta\mathcal{U}} \max_{u \in \mathcal{U}}\sum_{j \in \mathcal{U}} \beta(j) D_j^u.
\end{align}
%\begin{align}
%\label{maxmin}D^* &\doteq \max_{\vct{\alpha} \in \Delta\mathcal{U}} \min_{j \in \mathcal{U}}\alpha(u) D(p_0^u || p_1^u)\\
%\label{minmax}&=  \min_{\vct{\beta} \in \Delta\mathcal{U}} \max_{u \in \mathcal{U}} \beta(u) D(p_0^u || p_1^u).
%\end{align}
Note that $\alpha$ and $\beta$ are distributions over the set of experiments/components $\mathcal{U}$. Let $\alpha^*$ be the distribution that achieves the maximum in \eqref{maxmin} and let $\beta^*$ be the distribution that achieves the minimum in \eqref{minmax}. The equality of the min-max and max-min values follows from the minimax theorem \cite{osborne1994course} because the set $\mathcal{U}$ is finite and the Kullback-Leibler divergences are bounded due to Assumption \ref{boundedassump}.
\end{definition}

\begin{lemma}\label{kldexp}
The max-min Kullback-Leibler Divergence is equal to
\begin{align}
D^* = \left( \sum_{u \in \mathcal{U}} \frac{1}{D_u^u}\right)^{-1}.
\end{align}
And the distributions $\alpha^*$ and $\beta^*$ are given by $\alpha^*(u)= \beta^*(u) = D^*/D_u^u.$
%\begin{align}
%\alpha^*(u)= \beta^*(u) = D^*/D_u^u.
%\end{align}
\end{lemma}
\begin{proof}
See Appendix \ref{kldexpproof}.
\end{proof}

\begin{definition}\label{conflevel}
For given instance of information $\iota_{N+1} = \{u_{1:N},y_{1:N}\}$ of $I_{N+1}$, the we define \emph{confidence level} as
\begin{align}
\mathcal{C}(\iota_{N+1},\rho_1) \doteq \log \frac{\prod_{n = 1}^N p^{u_n}_0(y_n)}{\sum_{j \in \mathcal{U}}\tilde{\rho}_1(j) \prod_{n = 1}^N p^{u_n}_{\mathbbm{1}(u_n=j)}(y_n)},
\end{align}
where $\tilde{\rho}_1(j) = \rho_1(j)/(1-\rho_1(0)).$
\end{definition}
We will discuss more properties of this confidence level in Section \ref{analysis}. Note that this definition of confidence is consistent with the one in \cite{kartikarxiv}.

\subsection{Asymptotic Results for Heterogeneous Systems}
In this section, we state some results on the asymptotic behavior of the optimal value $\phi_N^*$ in Problem \eqref{opti}. The results stated here have been proven in \cite{kartikarxiv} in a setting where the observation space $\mathcal{Y}$ is finite. These results can be extended to infinite observation spaces in a fairly straightforward manner and thus, we omit their proofs.
\begin{lemma}[Weak Converse]\label{cslb}
The optimum value $\phi_N^*$ in Problem \eqref{opti} satisfies
\begin{align}
-\frac{1}{N}\log\phi^*_N 
 &\leq \frac{D^*}{1-\epsilon_N} + \frac{\log 2 + H(\beta^{*},\tilde{\rho}_1)}{N(1-\epsilon_N)}.\label{supbound}
\end{align}
\end{lemma}

\begin{assumption}\label{epsassum}
We have that the bound $1-\epsilon_N$ on the correct-verification probability in \eqref{opti} {satisfies $\epsilon_N \to 0$}. Further,
\begin{align}
    \lim_{N \to \infty}\frac{-\log{\epsilon_N}}{N} = 0.
\end{align}
\end{assumption}

\begin{assumption}\label{subgauss}
There exists a constant $B >0$ such that for each experiment $u$ and component $j$, the log-likelihood ratios $|\lambda_j(u,Y)| < B$.
\end{assumption}

Under Assumptions \ref{epsassum} and \ref{subgauss}, we can prove that for each $N$, there exist inference and experiment selections strategies $f^N, g^N$ such that they satisfy the constraint on correct-verification probability in Problem \ref{opti} and the corresponding incorrect-verification probability $\phi_N$ decays exponentially with $N$ at rate $D^*$. This achievability argument leads us to the following theorem.
\begin{theorem}
We have
\begin{align}
\lim_{N \to \infty}-\frac{1}{N}\log\phi_N^* = D^*.
\end{align}

\end{theorem}

\begin{remark}
One construction of inference and experiment selection strategies that achieve asymptotic optimality is the following. At each time $n$, randomly select a component $u$ with probability $\alpha^*(u)$. Then at time $N+1$, if the confidence $\mathcal{C}(I_{N+1},\tilde{\rho}_1)$ exceeds a threshold $\theta_N$ (precisely defined in \cite{kartikarxiv}), declare the system safe, and unsafe otherwise. Note that this experiment selection strategy is open-loop and randomized. An approach to design asymptotically optimal \emph{deterministic} and \emph{adaptive} experiment selection strategies is presented in \cite{kartikarxiv}.
\end{remark}

\subsection{Non-asymptotic Results for Homogeneous Systems}
We will now state our more recent non-asymptotic results for homogeneous systems.
The Kullback-Leibler divergence $D_u^u$ in homogeneous systems does not depend on the component $u$ and thus, we simply denote it by $D$. Using Lemma \ref{kldexp}, we can conclude that $D^* = D/M$, and that $\alpha^*$ and $\beta^*$ are uniform distributions over the set of components $\mathcal{U}$.

\begin{definition}\label{zdef}
For a given experiment selection strategy $g$ and a component $j \in \mathcal{U}$, define the total log-likelihood ratio up to time $n$ as
\begin{align*}
\rv{Z}_n(j) \doteq  \sum_{k=1}^n\lambda_j(\rv{U}_k,\rv{Y}_k),
\end{align*}
where the log-likelihood ratio $\lambda_j$ is as defined in equation (\ref{llrdef}). Also, let
\begin{align*}
L_n &\doteq \sum_{j \in \mathcal{U}}\beta^*(j)\lambda_j(U_n,Y_n)\\
\bar{Z}_n &\doteq  \sum_{j \in \mathcal{U}}\beta^{*}(j)Z_n(j) = \sum_{k = 1}^n L_k,
\end{align*}
where $\beta^{*}$ is the min-maximizing distribution in Definition \ref{kldef}.
\end{definition}

Since the system is homogeneous, we have
\begin{align}
L_n  &= \frac{1}{M}\sum_{j \in \mathcal{U}}\lambda_j(U_n,Y_n)= \frac{1}{M}\log\frac{p_0(Y_n)}{p_1(Y_n)}.
\end{align}
Notice that given $X=0$, the distribution of $L_n$ does not depend on $U_n$ and thus, on the strategy $g$.

\begin{lemma}\label{indeplemma}
When the system is homogeneous, $L_n$ is an i.i.d. sequence and the process $\bar{Z}_n$ is the sum of these i.i.d. random variables.
\end{lemma}
\begin{proof}
See Appendix \ref{indeplemmaproof}.
\end{proof}

Let $\textsc{inv}_n$ denote the quantile function (which is the same as the inverse-cdf if it exists) associated with the random varible $\bar{Z}_n + D(\beta^{*} || \tilde{\rho}_1)$.

\begin{theorem}\label{mainthm}
For homogeneous systems, we have
\begin{align}
\label{lb1}-\log{\phi^*_N} & \leq \textsc{inv}_N\left(\epsilon_N + \frac{\epsilon_N}{\eta}\right) + \log \frac{\eta}{\epsilon_N}\\
\label{ub1}-\log{\phi^*_N} & \geq \textsc{inv}_N\left(\epsilon_N - \frac{\epsilon_N}{\eta}\right) - O\left(\log \frac{\eta}{\epsilon_N}\right),
\end{align}
for any $\eta > 1$ as long as the argument of $\textsc{inv}_N \in (0,1)$. Note that $\eta$ may also depend on $N$.
\end{theorem}
\begin{remark}
The bound in \eqref{ub1} is stated in big-O notation because the constants associated with the logarithmic term are difficult to determine in general. Herein, we only prove the \emph{existence} of constants that achieve the bound \eqref{ub1}.
\end{remark}
We would like to emphasize that Theorem \ref{mainthm} does not require Assumptions \ref{epsassum} and \ref{subgauss}. The bound in \eqref{lb1} is based on the strong converse theorem \cite{polyathesis} and other properties the log-likelihood ratios $Z_n(j)$. The result in \eqref{ub1} is obtained by constructing experiment selection and inference strategies and bounding their performance. The approach used for bounding performance is based on the well-known Chernoff-bound \cite{ross2014introduction}. The experiment selection strategy that achieves the bound is as follows: at each time $n$, select the component $j$ that minimizes $Z_{n-1}(j) - \log \tilde{\rho}_1(j)$. The inference strategy is a simple confidence-threshold based strategy. Detailed descriptions of these strategies and the complete proof of Theorem \ref{mainthm} is provided in Section \ref{analysis}.

\begin{remark}\label{plotrem}
For a simple setup with two components and a binary observation space, we numerically examine the performance of our deterministic adaptive strategy (DAS, described above) and the open-loop randomized strategy (ORS) which selects each component with probability 0.5.. Figure \ref{fig:plot1} depicts the results of these numerical experiment and also, the weak \eqref{supbound} and strong \eqref{lb1} converse bounds. Notice that the performance of the deterministic strategy is very close to the strong converse indicating the tightness of the converse bound and efficiency of the designed strategy. More details on these results are provided in \cite{kartikarxiv}.
\end{remark}

\begin{figure}[h]
\centering
\includegraphics[width = \columnwidth]{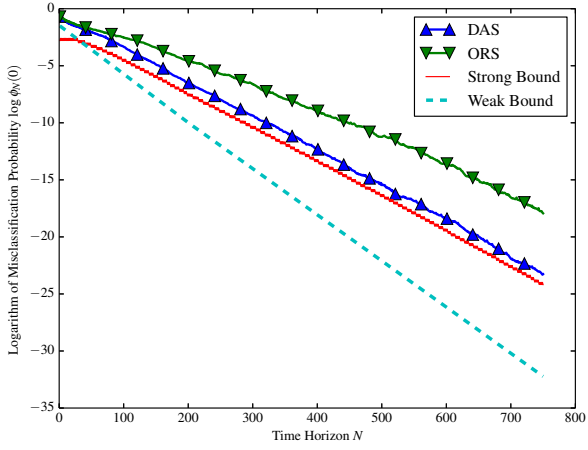}
\caption{The plot depicts the converse bounds and the performance of the experiment selection strategies mentioned in Remark \ref{plotrem}.}
\label{fig:plot1}
\end{figure}

 Since $L_n$ is an i.i.d. collection of random variables, we can further simplify these bounds by approximating the quantile function $\textsc{inv}_N$ using the Berry-Esseen Theorem \cite{polyathesis}.

\begin{corollary}[Berry-Esseen]\label{bess}
If $V \doteq \E_0[(L_1 - D^*)^2]$ and $T \doteq \E_0[|L_1 - D^*|^3] < \infty$, then
\begin{align}
\label{lb2}&-\log{\phi^*_N} \leq \\
\nonumber & {ND^*} - {\sqrt{NV}}Q^{-1}\left(\epsilon_N+\frac{\epsilon_N}{\eta} + \frac{6T}{\sqrt{NV^3}}\right) +O\left(\log\frac{\eta}{\epsilon_N} \right),\\
%\shortintertext{and}
\label{ub2}&-\log{\phi^*_N} \geq \\
\nonumber & {ND^*} - {\sqrt{NV}}Q^{-1}\left(\epsilon_N-\frac{\epsilon_N}{\eta} - \frac{6T}{\sqrt{NV^3}}\right) -O\left(\log\frac{\eta}{\epsilon_N} \right).
\end{align}
Here, the $Q$-function is the tail distribution function of the standard normal distribution. The results above are valid only when the argument of $Q^{-1}$ is between 0 and 1.
\end{corollary}
%The Berry-Esseen approximation leads to tight bounds when $\epsilon_N$ does not decay too fast like when $\epsilon_N$ is a constant or when $\epsilon_N = \Theta(\log N /\sqrt{N})$.

\section{Outline of Proof of Theorem \ref{mainthm}}\label{analysis}

We will first define some important quantities that will be used in our analysis. Let us arbitrarily fix the experiment selection strategy to be some $g \in \mathcal{G}$. For a given instance $\iota_{N+1} = \{u_{1:N},y_{1:N} \}$ of information, let $\iota_{n+1} =  \{u_{1:n},y_{1:n} \}$ for $n \leq N$. Define
\begin{align}
P^g_{N}(\iota_{N+1}) &\doteq \prod_{n = 1}^N \Py^g[u_n \mid \iota_n]p^{u_n}_0(y_n)\\
Q^g_{N}(\iota_{N+1}) &\doteq \sum_{j \in \mathcal{U}}\tilde{\rho}_1(j) \prod_{n = 1}^N \Py^g[u_n \mid \iota_n]p^{u_n}_{\mathbbm{1}(u_n=j)}(y_n).
\end{align}
The quantities $P^g_{N}$ and $Q_{N}^g$ are densities over the space $\mathcal{I}_{N+1}$ conditioned on $X=0$ and $X \neq 0$ respectively. The densities are with respect to the product measure $\Bbbk^N \times \nu^N$ where $\Bbbk$ denotes the counting measure. 
%Therefore, for measurable set $\I_{N+1} \subseteq \mathcal{I}_{N+1}$, we have
%\begin{align*}
%\Py^g[I_{N+1} \in \I_{N+1} \mid X = 0] &= \int\displaylimits_{\I_{N+1}}P^g_{N}(\iota) d(\Bbbk^N \times \nu^N)(\iota) \\
%\Py^g[I_{N+1} \in \I_{N+1} \mid X \neq 0] &= \sum_{j \in \mathcal{U}}\tilde{\rho}_1(j) \Py^g_j[I_{N+1} \in \I_{N+1}]\\
%&= \int\displaylimits_{\I_{N+1}}Q^g_{N}(\iota) d(\Bbbk^N \times \nu^N)(\iota).
%\end{align*}
Clearly, for any strategy $g$,
\begin{align}
\mathcal{C}(I_{N+1},\rho_1) = \log\frac{P^g_{N}(I_{N+1})}{Q^g_{N}(I_{N+1})}
\end{align}
with probability 1. 

The first step in analyzing Problem \eqref{opti} with the experiment selection strategy ($g$) fixed is to view it as a one-shot binary hypothesis testing problem in the following manner: $N$ experiments are performed using the strategy $g$ and then we observe the sequence $I_{N+1}$ of experiments and observations. This observed sequence $I_{N+1}$ can be viewed as a \emph{single observation}. Based on $I_{N+1}$, we need to infer whether the system is safe ($X = 0$) or it is unsafe ($X \neq 0$) using the inference strategy $f$. If $X = 0$, the density associated with the observed sequence $I_{N+1}$ is $P_N^g$ and if $X \neq 0$, then the density associated with the sequence $I_{N+1}$ is $Q_N^g$. Thus, $\mathcal{C}(I_{N+1}, \rho_1)$ is the log-likelihood ratio associated with this one-shot binary hypothesis testing problem. We can now use the strong converse (or the weak converse) \cite{polyanotes} to obtain a lower bound on $\phi_N^*$. However, this bound will depend on the choice of strategy $g$.

%Notice that the expected increment in the process $\bar{Z}_n$ is given by
%\begin{align}
%\nonumber \E_i^g[\bar{Z}_n - \bar{Z}_{n-1} \mid I_n] &=\E_i^g\left[\sum_{j\neq i}\beta^{i*}(j)\lambda_j^i(U_n,Y_n) \mid I_{n}\right]\\
%& = \sum_{j\neq i}\beta^{i*}(j)D(p^{U_n}_i||p^{U_n}_j)\\&\geq 0,
%\end{align}
%with probability 1. Thus, $\bar{Z}_n$ is a sub-martingale with respect to the filtration $I_{n+1}$.
We can exploit various properties of the confidence level $\mathcal{C}$ to obtain a strategy-independent lower bound on $\phi_N^*$. To do so, we will first establish the relationship between the total log-likelihood ratios $Z_n(j)$ and the confidence level $\mathcal{C}$.

\begin{lemma}\label{logsumexplemma}
For any experiment selection strategy $g$ and for each $1 \leq n \leq N$, we have
\begin{align}
\nonumber\mathcal{C}(I_{n+1},\vct{\rho}_{1}) =  - \log\left[\sum_{j \in \mathcal{U}}\exp\Big(\log\tilde{\rho}_1(j) -  Z_n(j)\Big)\right],
\end{align}
where $\tilde{\rho}_1(j) = \rho_1(j)/(1-\rho_1(0)).$
\end{lemma}
\begin{proof}
The proof of this lemma can be found in \cite{kartikarxiv}.
\end{proof}

We use Lemma  \ref{logsumexplemma} to decompose the confidence into a Kullback-Leibler divergence term and a sub-martingale (a simple i.i.d. sum in the case of homogeneous systems).

\begin{lemma}[Decomposition]\label{decom}
For any experiment selection strategy $g$, we have
\begin{align*}
\mathcal{C}(I_{n+1},\vct{\rho}_{1})  =\Big[-D(\beta^{*}||\tilde{\rho}_{n+1})\Big] +\Big[ \bar{Z}_n + D(\beta^{*} || \tilde{\rho}_1) \Big],
\end{align*}
%\begin{align}
%&\nonumber\mathcal{C}_i(\vct{\rho}_{n+1}) - \mathcal{C}_i(\vct{\rho}_1) \\
%%&=  - \log\left[\sum_{j\neq i}\exp\left(\log\tilde{\rho}_1(j) -Z_n(j)\right)\right]\\
%\nonumber&= - \log\sum_{j\neq i}\exp\left(\log\tilde{\rho}_1(j) -Z_n(j) + \bar{Z}_n + H(\beta^{i*},\tilde{\rho}_1)\right)\\
%&\quad \,+\bar{Z}_n + H(\beta^{i*},\tilde{\rho}_1)\\
%&= -H(\beta^{i*},\tilde{\rho}_{n+1}) + \bar{Z}_n + H(\beta^{i*},\tilde{\rho}_1).
%\end{align}
where $$\tilde{\rho}_{n+1}(j) = \frac{\tilde{\rho}_1(j)e^{-Z_n(j)}}{\sum_{k \in \mathcal{U}}\tilde{\rho}_1(k)e^{-Z_n(k)}}.$$
\end{lemma}
\begin{proof}
See Appendix \ref{decomproof}.
\end{proof}

The decomposition in Lemma \ref{decom} allows us to obtain the bounds in Theorem \ref{mainthm}. We can analyze the two terms in the decomposition separately. The key property to exploit is that $\bar{Z}_n$ is a sum of i.i.d. variables and does not depend on the strategy $g$. Further, for the strong converse, we simply exploit the non-negativity of Kullback-Leibler divergence. For the achievability bound, we can show that the adaptive strategy keeps the Kullback-Leibler divergence term small by means of a Chernoff bound. The details of this proof are provided in Appendix \ref{compproof}.

\section{Conclusions}
In this paper, we considered a fixed-horizon Neyman-Pearson type formulation for the problem of actively verifying whether a multi-component system is free of anomalies or not. We studied the asymptotics of the problem and for the specific class of homogeneous problems, we provided stronger non-asymptotic converse and performance bounds. We observed that the strong converse is fairly tight. The strong achievability bounds are order-optimal up to logarithmic terms but the constants associated with the logarithmic terms are not known. These constants remain to be analyzed for future work.

\section*{Acknowledgment}
This research was supported, in part, by National Science Foundation under Grant NSF CCF-1817200, CCF-1718560, CPS-1446901, Grant ONR N00014-15-1-2550, and Grant ARO W911NF1910269.

%%%%%%
%% Appendix:
%% If needed a single appendix is created by
%%
%\appendix
%%
%% If several appendices are needed, then the command
%%
% \appendices
%%
%% in combination with further \section-commands can be used.
%%%%%%

\newpage

\bibliographystyle{IEEEtran}
\bibliography{refs}

\clearpage

\appendices
\section{Proof of Lemma \ref{kldexp}}\label{kldexpproof}
We have
\begin{align}
D^* = \max_{\vct{\alpha} \in \Delta\mathcal{U}} \min_{j \in \mathcal{U}} \sum_{u \in \mathcal{U}}\alpha(u) D_j^u = \max_{\vct{\alpha} \in \Delta\mathcal{U}} \min_{j \in \mathcal{U}} \alpha(j) D_j^j,
\end{align}
because $D_j^u = 0$ if $j \neq u$. Based on a simple contradiction argument, we can conclude that for every $u \in \mathcal{S}$, the max-minimizer $\alpha^*$ must satisfy
\begin{align}
\alpha^*(u)D(p_0^u||p_1^u) &= D^*.
\end{align}
This, combined with the fact that $\alpha^*$ is a distribution over $\mathcal{S}$, leads us to the result. The distribution $\beta^*$ can be obtained in a similar manner.

\section{Proof of Lemma \ref{indeplemma}}\label{indeplemmaproof}
Let the moment generating function of $L_n$ be $\bar{\mu}(s)$. Therefore, for any strategy $g$, we have
\begin{align}
&\E_0^g[\exp(\sum_{k =1}^ns_kL_k)] = \E_0^g[\E_0\exp(\sum_{k =1}^ns_kL_k)\mid I_n]]\\
&=\nonumber \E_0^g[\exp(\sum_{k =1}^{n-1}s_kL_k)\E_0[\exp(s_nL_n)\mid I_n]]\\
&= \E_0^g[\exp(\sum_{k =1}^{n-1}s_kL_k)]\bar{\mu}(s_n) =\Pi_{k=1}^n\bar{\mu}(s_k).
\end{align}

 \section{Proof of Lemma \ref{decom}}\label{decomproof}
Using the definition of cross-entropy, we have
\begin{align*}
&-H(\beta^{*},\tilde{\rho}_{n+1}) = \sum_{j \in \mathcal{U}}\beta^{*}(j)\log\tilde{\rho}_{n+1}(j)\\
&\stackrel{a}{=}\sum_{j \in \mathcal{U}}\beta^{*}(j)\log\left(\frac{\tilde{\rho}_1(j)e^{-Z_n(j)}}{\sum_{k\in \mathcal{U}}\tilde{\rho}_1(k)e^{-Z_n(k)}} \right)\\
&= -H(\beta^{*},\tilde{\rho}_1) - \bar{Z}_n - \log\left[\sum_{k\in \mathcal{U}}\exp\left(\log\tilde{\rho}_1(k) -  Z_n(k)\right)\right]\\
&\stackrel{b}{=}-H(\beta^{*},\tilde{\rho}_1) - \bar{Z}_n + \mathcal{C}(I_{n+1},\vct{\rho}_1).
\end{align*}
Equality $(a)$ follows from the definition of $\tilde{\rho}_{n+1}$ and equality $(b)$ is a consequence of Lemma \ref{logsumexplemma}. The lemma then follows by adding $H(\beta^*,\beta^*)$ on both sides.

\section{Complete Proof}\label{compproof}

\subsection{Strong Converse}
For any given pair of inference and experiment selection strategies $f,g$ that are feasible in Problem (\ref{opti}), the confidence level $\mathcal{C}$ can be viewed as a log-likelihood ratio. Therefore for this strategy pair $f,g$, we have the following for every $\chi \in \R$
\begin{align}
&-\log{\phi_N}\\
 &\stackrel{a}{\leq}\, \chi - \log(\psi_N - \Py_0^g[\mathcal{C}(I_{N+1},\rho_1) > \chi])\\
& \stackrel{b}{\leq} \, \chi - \log(1-\epsilon_N - \Py_0^g[\mathcal{C}(I_{N+1},\rho_1) > \chi])\\
&= \, \chi - \log(\Py_0^g[\mathcal{C}(I_{N+1},\rho_1)  \leq \chi]-\epsilon_N )\label{strong3}
\end{align}
Here, we use the convention that if $x \leq 0$, then $\log x \doteq -\infty$. Inequality $(a)$ is a consequence of the strong converse theorem in \cite{polyanotes}. Inequality $(b)$ holds because $\psi_N \geq 1-\epsilon_N$. However, this lower bound on $\phi_N$ depends on the experiment selection strategy $g$. We can use the decomposition in Lemma \ref{decom} to obtain a strategy-independent lower bound.
We have
\begin{align}
&\Py^g_0[\mathcal{C}(I_{N+1},\rho_1)  \leq \chi] \\
&\stackrel{a}{=} \Py^g_0[-D(\beta^{*}||\tilde{\rho}_{N+1})+\bar{Z}_N + D(\beta^{*}||\tilde{\rho}_1) \leq \chi]\\
&\stackrel{b}{\geq} \Py_0[\bar{Z}_N + D(\beta^{*}||\tilde{\rho}_1) \leq \chi].\label{strong1}
\end{align}
Equality $(a)$ is a consequence of Lemma \ref{decom}, and since $D(\beta^{*}||\tilde{\rho}_{N+1}) \geq 0$, we have that the event
\begin{align}
\{-D(\beta^{*}||\tilde{\rho}_{N+1})+\bar{Z}_N + D(\beta^{*}||\tilde{\rho}_1) \leq \chi\} \\
\supseteq \{\bar{Z}_N + D(\beta^{*}||\tilde{\rho}_1) \leq \chi\},\label{strong2}
\end{align}
which results in the inequality $(b)$. Combining \eqref{strong3} and \eqref{strong1} leads us to the following lemma.
\begin{lemma}[Stong Converse]\label{strongconv}
For any given pair of inference and experiment selection strategies $f,g$ that are feasible in Problem (\ref{opti}), we have for every $\chi \in \R$
\begin{align*}
-\log{\phi_N} \leq \chi - \log(\Py^g_0[\bar{Z}_N + D(\beta^{*}||\tilde{\rho}_1) \leq \chi]-\epsilon_N ),
\end{align*}
with the convention that $\log x \doteq -\infty$ if $x \leq 0$.
\end{lemma}
The bound \eqref{lb1} in Theorem \ref{mainthm} is obtained from Lemma \ref{strongconv} by assigning
\begin{align}
\chi =  \textsc{inv}_N\left(\epsilon_N + \frac{\epsilon_N}{\eta}\right),
\end{align}
where $\textsc{inv}_N$ is the quantile function of $\bar{Z}_N + D(\beta^{*}||\tilde{\rho}_1)$.
 The bound \eqref{lb2} in Corollary \ref{bess} is obtained by assigning
\begin{align*}
\chi^*  \doteq {ND^*} - {\sqrt{NV}}Q^{-1}\left(\epsilon_N+\frac{\epsilon_N}{\eta} + \frac{6T}{\sqrt{NV^3}}\right)+ D(\beta^{*}||\tilde{\rho}_1).
\end{align*}
Then using the Berry-Esseen theorem \cite{polyathesis}, we have
\begin{align}
\Py^g_0[\bar{Z}_N + D(\beta^{*}||\tilde{\rho}_1) \leq \chi^*] \geq \epsilon_N+\frac{\epsilon_N}{\eta}.
\end{align}
The result above combined with Lemma \ref{strongconv} leads to the strong converse bound in Corollary \ref{bess}.

\subsection{Strong Achievability}
\begin{lemma}\label{thresh}
Let $f$ be a deterministic inference strategy in which hypothesis $0$ is decided only if $\mathcal{C}(I_{N+1},\rho_1) \geq \theta$. Then $\phi_N \leq e^{-\theta}.$
%\begin{equation}
%   \phi_N(i) \leq e^{-\theta}.
%\end{equation}
\end{lemma}
\begin{proof}
This follows from the standard arguments associated with log-likelihood ratios \cite{polyanotes}. A proof of this lemma is provided in Appendix G of \cite{kartikarxiv} for the case when the observation space $\mathcal{Y}$ is finite.
\end{proof}

Consider the following inference strategy: decide that the system is safe if $\mathcal{C}(I_{N+1},\rho_1) \geq \theta_N$ and decide that it is unsafe otherwise, where the threshold $\theta_N$ is given by
\begin{align}
\theta_N \doteq \textsc{inv}_N\left(\epsilon_N - \frac{\epsilon_N}{\eta}\right) - O\left(\log \frac{\eta}{\epsilon_N}\right).
\end{align}
Using Lemma \ref{thresh}, we can conclude that the inference strategy stated above (irrespective of which experiment selection strategy is used) achieves $\phi_N \leq \exp(-\theta_N)$. However, for a pair of experiment selection and inference strategies to be feasible, we also need to show that the constraint ($\psi_N \geq 1-\epsilon_N$) in Problem \eqref{opti} is satisfied. To do so, all we need to show is that our deterministic adaptive strategy, combined with the threshold based inference described above, satisfies
\begin{align}
\Py_0^g[\mathcal{C}(I_{N+1},\rho_1) < \theta_N] \leq \epsilon_N.
\end{align}
Based on the decomposition in Lemma \ref{decom} and a union bound, a sufficient criterion for the satisfying condition above is the following:
\begin{align}
\label{acond1}&\Py_0^g[-D(\beta^{*}||\tilde{\rho}_{n+1}) < \theta_{N,1}] \leq \epsilon_N /\eta\\
\label{acond2}&\Py_0^g[\bar{Z}_n + D(\beta^{*} || \tilde{\rho}_1) < \theta_{N,2}] \leq \epsilon_N - \epsilon_N/\eta,
\end{align}
where $\eta > 1$ and
\begin{align}
\theta_{N,1} &=- O\left(\log \frac{\eta}{\epsilon_N}\right)\\
\theta_{N,2} &= \textsc{inv}_N\left(\epsilon_N - \frac{\epsilon_N}{\eta}\right).
\end{align}
Notice that the condition \eqref{acond2} is trivially satisfied because of the definition of the $\theta_{N,2}$ and the quantile function $\textsc{inv}_N$. Therefore, we just need to show that condition \eqref{acond1} is satisfied. To do so, we will use a Chernoff bound based argument in the following manner.

\begin{remark}
The result \eqref{ub2} in Corollary \ref{bess} can be obtained by assigning
\begin{align}
&\theta_{N,2} \doteq \\
&{ND^*} - {\sqrt{NV}}Q^{-1}\left(\epsilon_N-\frac{\epsilon_N}{\eta} - \frac{6T}{\sqrt{NV^3}}\right) + D(\beta^{*} || \tilde{\rho}_1) .
\end{align}
Once again, using the Berry-Esseen theorem \cite{polyathesis}, 
\begin{align}
\Py_0^g[\bar{Z}_n + D(\beta^{*} || \tilde{\rho}_1) < \theta_{N,2}] \leq \epsilon_N - \epsilon_N/\eta.
\end{align}
\end{remark}

\paragraph*{Proof of Inequality \eqref{acond1}}
Define
\begin{align}
\zeta_n(j) \doteq Z_n(j) - \log \tilde{\rho}_1(j) - \bar{Z}_n - H(\beta^*,\tilde{\rho}_1).
\end{align}
In the homogeneous case, we have
\begin{align}
\zeta_n(j) - \zeta_{n-1}(j) &= \lambda_j(U_n,Y_n) - \frac{1}{M}\sum_{k \in \mathcal{U}}\lambda_j(U_n,Y_n)\\
&= \begin{cases}
\frac{M-1}{M}\log\frac{p_0(Y_n)}{p_1(Y_n)} &\text{if } U_n = j\\
\frac{1}{M}\log\frac{p_1(Y_n)}{p_0(Y_n)} &\text{if } U_n \neq j\
\end{cases}
\end{align}
For convenience, define
\begin{align}
m(s) &= \E_0\left[\exp\left(\frac{s(M-1)}{M}\log\frac{p_1(Y)}{p_0(Y)}\right)\right]\\
\bar{m}(s) &= \E_0\left[\exp\left(\frac{s}{M}\log\frac{p_0(Y)}{p_1(Y)}\right)\right].
\end{align}

Notice that the strategy DAS described earlier is equivalent to selecting the component $j$ with least $\zeta_n(j)$ at time $n+1$. Let
\begin{align}
\bar{j}_n \doteq \argmin_{j \in \mathcal{U}} \zeta_n(j).
\end{align}
We have
\begin{align}
 -H(\beta^{*},\tilde{\rho}_{n+1}) = -\log\sum_{j \in \mathcal{U}} \exp(-\zeta_n(j))
\end{align}
Therefore, for $0 \leq s \leq 1$,
\begin{align}
\nonumber\E^g_0\exp[sH(\beta^{*},\tilde{\rho}_{n+1})] &= \E^g_0\left(\sum_{j \in \mathcal{U}}\exp\left( -\zeta_n(j)\right)\right)^s\\
\nonumber &= \E^g_0\left(\sum_{j \in \mathcal{U}}(\exp\left(-s\zeta_n(j)\right))^{1/s}\right)^s\\
&\stackrel{a}{\leq}\E^g_0\left[\sum_{j\in \mathcal{U}}\exp\left(-s\zeta_n(j)\right)\right]\\
&= \sum_{j\in \mathcal{U}}\E^g_0\exp\left(-s\zeta_n(j)\right).\label{normineq2}
\end{align}
Inequality $(a)$ holds because $\| \cdot \|_{1/s} \leq \| \cdot \|_1.$
Further, we have
\begin{align}
&\sum_{j\in \mathcal{U}}\E^g_0\left[\exp\left(-s\zeta_{n+1}(j)\right)\right]\\
&\nonumber\stackrel{a}{=} \E^g_0\left[\sum_{j\in \mathcal{U}}\exp\left(-s\zeta_n(j)\right)\E^g_0[\exp(-s(\zeta_{n+1}(j) - \zeta_n)) \mid I_{n+1}]\right]\\
&= \E_0^g\left[\exp(-s\zeta_n(\bar{j}_n))m(s) + \sum_{j \neq \bar{j}_n} \exp(-s\zeta_n(j))\bar{m}(s)\right],\label{normineq3}
\end{align}
where $(a)$ follows from the tower property of conditional expectation and the fact that $\zeta_n(j)$ is measurable w.r.t. $I_{n+1}$. Let
\begin{align}
\varrho_n \doteq \frac{ \exp(-s\zeta_n(\bar{j}_n))}{\sum_{j\in \mathcal{U}}\exp(-s\zeta_n(j))} \geq \frac{1}{M}.
\end{align}
Then
\begin{align}
&\frac{ \exp(-s\zeta_n(\bar{j}_n))m(s) + \sum_{j \neq \bar{j}_n} \exp(-s\zeta_n(j))\bar{m}(s)}{\sum_{j\in \mathcal{U}}\exp(-s\zeta_n(j))}\\
&= \varrho_nm(s) + (1-\varrho_n)\bar{m}(s).
\end{align}

\begin{lemma}\label{achieve1}
For $0 \leq \delta < 1/(M-1)$, if $\varrho_n \leq \frac{1+\delta}{M}$ then
\begin{align}
 \exp(-s\zeta_n(\bar{j}_n)) = \max_{j \in \mathcal{U}} \exp(-s\zeta_n(j)) \leq \frac{1+\delta}{1+\delta -M \delta}.
\end{align}
\end{lemma}
\begin{proof}
Consider the following facts:
\begin{enumerate}
\item  We have $\min_{j \in \mathcal{U}}\exp(-s\zeta_n(j)) \leq 1$. This is because $\sum_{j\in \mathcal{U}} \zeta_n(j) = 0$ and thus $\max_{j \in \mathcal{U}} \zeta_n(j) \geq 0$.
\item For every $j \in \mathcal{U}$, we have $$\exp(-s\zeta_n(j)) \leq  \exp(-s\zeta_n(\bar{j}_n)).$$ This simply follows from the definition of $\bar{j}_n$.
\end{enumerate}
Combining the two facts stated above, we have
$$\sum_{j\in \mathcal{U}}\exp(-s\zeta_n(j)) \leq (M-1)\exp(-s\zeta_n(\bar{j}_n)) + 1. $$
Therefore,
\begin{align}
 \frac{1+\delta}{M} &\geq \frac{ \exp(-s\zeta_n(\bar{j}_n))}{\sum_{j\in \mathcal{U}}\exp(-s\zeta_n(j))}\\
 & \geq \frac{\exp(-s\zeta_n(\bar{j}_n))}{(M-1)\exp(-s\zeta_n(\bar{j}_n)) + 1}\\
\implies  & \max_{j \in \mathcal{U}} \exp(-s\zeta_n(j)) \leq \frac{1+\delta}{1+\delta -M \delta}
\end{align}
This concludes the proof of the lemma.
\end{proof}

\begin{lemma}\label{achieve2}
There exist constants $0 < s^* < 1$ and $0 < \varsigma < 1$ such that if $\varrho_n > \frac{1+\delta}{M}$, then $\varrho_nm(s^*) + (1-\varrho_n)\bar{m}(s^*) < \varsigma.$
\end{lemma}
\begin{proof}
Based on a first-order Taylor approximation at $s = 0$, we can conclude that
\begin{align}
m(s) &= 1 - s\frac{(M-1)D}{M} + o(s)\\
\bar{m}(s) &= 1 + s\frac{D}{M} + o(s).
\end{align}
Therefore, there exists a neighborhood of $s$ around 0 such that $m(s) < \bar{m}(s)$. Hence, in this neighborhood $\varrho_nm(s) + (1-\varrho_n)\bar{m}(s)$ is a decreasing function of $\varrho_n$ for fixed value of $s$. Further,
\begin{align}
\varrho_nm(s) + (1-\varrho_n)\bar{m}(s) = 1 + s\left(\frac{D}{M} -\varrho_n\right) + o(s). 
\end{align}
For $\varrho_n = \frac{1+\delta}{M}$, the RHS in the expression above is
\begin{align}
1 - s\left(\frac{D\delta}{M}\right) + o(s).
\end{align}
Therefore, there exists an $s^*$ such that
\begin{align}
\varsigma  \doteq \frac{1+\delta}{M}m(s^*) + \left(1-\frac{1+\delta}{M}\right)\bar{m}(s^*) < 1.
\end{align}
Hence, for every $\varrho_n > \frac{1+\delta}{M}$, $\varrho_nm(s^*) + (1-\varrho_n)\bar{m}(s^*) < \varsigma < 1.$ This concludes the proof of the lemma.
\end{proof}
Henceforth, the value of $s$ is assigned to be $s^*$ defined in the proof of Lemma \ref{achieve2}. Based on Lemmas \ref{achieve1} and \ref{achieve2}, we can consider the the following cases:
\begin{enumerate}
\item  $\varrho_n \leq \frac{1+\delta}{M}$ : In this case, we can conclude using Lemma \ref{achieve1} that
\begin{align}
\exp(-s\zeta_n(\bar{j}_n))m(s) + \sum_{j \neq \bar{j}_n} \exp(-s\zeta_n(j))\bar{m}(s)\\
 \leq \frac{(1+\delta)(m(s) + (M-1)\bar{m}(s))}{1+ \delta - M\delta} \doteq K.
\end{align}
\item $\varrho_n > \frac{1+\delta}{M}$: In this case, we can conclude using Lemma \ref{achieve2} that
\begin{align}
\exp(-s\zeta_n(\bar{j}_n))m(s) + \sum_{j \neq \bar{j}_n} \exp(-s\zeta_n(j))\bar{m}(s)\\
< \varsigma \left(\sum_{j\in \mathcal{U}}\exp(-s\zeta_n(j))\right).
\end{align}
\end{enumerate}
Therefore, we have
\begin{align}
\exp(-s\zeta_n(\bar{j}_n))m(s) + \sum_{j \neq \bar{j}_n} \exp(-s\zeta_n(j))\bar{m}(s)\\
< \max\{K, \varsigma \sum_{j\in \mathcal{U}}\exp(-s\zeta_n(j))\}.
\end{align}

This, using the result above and \eqref{normineq3}, we can conclude that
\begin{align}
\sum_{j\in \mathcal{U}}\E^g_0\exp\left(-s\zeta_{n+1}(j)\right)<K + \varsigma \sum_{j\in \mathcal{U}}\E_0^g\exp(-s\zeta_n(j)).\label{induct}
\end{align}

Using the result \eqref{induct} inductively and combining it with \eqref{normineq2}, we have
\begin{align}
\E^g_0\exp[sH(\beta^{*},\tilde{\rho}_{N+1})] &\leq  M\varsigma^{N} + \sum_{n=1}^{N} K\varsigma^{N-n}\\
& \leq M + \frac{K}{1-\varsigma} \doteq K'.\label{mgf1}
\end{align}

\paragraph*{Chernoff Bound}
We can use the Chernoff bound \cite{ross2014introduction} to conclude that
\begin{align}
&\Py_0^g[-D(\beta^{*}||\tilde{\rho}_{n+1}) < \theta_{N,1}]\\
&=\Py_0^g[\log M-H(\beta^{*},\tilde{\rho}_{N+1}) < \theta_{N,1}]\\
& \leq \exp(s(\theta_{N,1}-\log M))\E^g_0\exp[sH(\beta^{*},\tilde{\rho}_{N+1})] \\
&\stackrel{a}{\leq}  \exp(s(\theta_{N,1}-\log M))\left(M + \frac{K}{1-\varsigma}\right)\\
%&\stackrel{b}{=} M \exp \left( -\frac{(ND^*(i) - \theta)^2}{2NB^2}\right)\\
&\stackrel{b}{=} \epsilon_N/\eta,
\end{align}
where
\begin{align}
\theta_{N,1} \doteq \frac{1}{s^*}\log\frac{\epsilon_N}{\eta K'} + \log M =- O\left(\log \frac{\eta}{\epsilon_N}\right) .
\end{align}
This concludes our argument that the condition \eqref{acond1} is satisfied.

\begin{remark}
We would like to emphasize that the constants in the logarithmic term such as $\delta, s^*$ etc. need to be chosen appropriately and at this point it is not clear how one might determine these constants in general. It remains to be investigated how tight the bound obtained herein would be once these constants are obtained.
\end{remark}

%%%%%%
%% To balance the columns at the last page of the paper use this
%% command:
%%
%\enlargethispage{-1.2cm} 
%%
%% If the balancing should occur in the middle of the references, use
%% the following trigger:
%%
\IEEEtriggeratref{300}

\end{document}